\documentclass{article}

\usepackage{authblk}

\usepackage{amsmath} 
\usepackage{amsthm} 
\usepackage{amssymb} 
\usepackage{amsbsy} 
\usepackage{braket} 
\usepackage{graphicx} 
\usepackage{pstricks} 
\usepackage{color} 
\usepackage{enumerate} 
\usepackage{todonotes} 
\usepackage{comment} 
\usepackage{hyperref} 
\usepackage[capitalise]{cleveref} 

\newtheorem{thm}{Theorem}
\theoremstyle{definition}
\newtheorem{definition}[thm]{Definition}

\newtheorem{lem}[thm]{Lemma}

\theoremstyle{remark}
\theoremstyle{definition}
\newcommand{\beq}{\begin{equation}}
\newcommand{\eeq}{\end{equation}}
\newcommand{\ba}{\begin{eqnarray}}

\newcommand{\ea}{\end{eqnarray}}
\newcommand{\ban}{\begin{eqnarray*}}
\newcommand{\ean}{\end{eqnarray*}}
\newcommand{\bit}{\begin{itemize}}
\newcommand{\eit}{\end{itemize}}
\newcommand{\ben}{\begin{enumerate}}
\newcommand{\een}{\end{enumerate}}

\newcommand{\tr}{\ensuremath\mathrm{Tr}}

\begin{document}

\title{On modifications of Reichenbach's principle of common cause in light of Bell's theorem}
\author[1,2]{Eric G. Cavalcanti}
\affil[1]{\small School of Physics, University of Sydney, NSW 2016, Australia}
\author[2,1]{Raymond Lal}
\affil[2]{Department of Computer Science, University of Oxford, Oxford OX1 3QD, United Kingdom}

\date{June 16, 2014}

\maketitle

\begin{abstract}
Bell's 1964 theorem causes a severe problem for the notion that correlations require explanation, encapsulated in Reichenbach's Principle of Common  Cause. Despite being a hallmark of scientific thought, dropping the principle has been widely regarded as a much less bitter medicine than the perceived alternative---dropping relativistic causality. Recently, however, some authors have proposed that modified forms of Reichenbach's principle could be maintained even with relativistic causality. Here we break down Reichenbach's principle into two independent assumptions---the \emph{Principle of Common Cause} proper, and \emph{factorisation of probabilities}. We show how Bell's theorem can be derived from these two assumptions plus relativistic causality and the Law of Total Probability for actual events, and we review proposals to drop each of these assumptions in light of the theorem. In particular, we show that the \emph{non-commutative common causes} of Hofer-Szab\'{o} and Vecserny\'{e}s fail to have an analogue of the notion that the common causes can explain the observed correlations. Moreover, we show that their definition can be satisfied trivially by any quantum product state for any quantum correlations. We also discuss how the \emph{conditional states} approach of Leifer and Spekkens fares in this regard.
\end{abstract}

\section{Introduction}\label{sec:intro}

Imagine that a rare viral infection breaks out simultaneously in Australia and in Brazil. We would expect one of three possible explanations to account for this coincidence: either the virus was carried by a host travelling from Australia to Brazil, or by a host going from Brazil to Australia, or it was carried by two travellers who were in contact at some third country in the past. Reichenbach's Principle of Common Cause (R-PCC) \cite{Reichenbach1956} attempts to formalise this intuition. It states that when two events are correlated, either one is a direct cause of the other, or they share a common cause. This principle is considered to be fundamental for many areas of scientific practice, but it leads to severe problems when considered in conjunction with relativistic causal structure and the predictions of quantum theory, as shown by Bell's theorem \cite{Bell1964, Wood2012}. Here we will discuss some attempts to retain some form of the principle. We will demonstrate that one of those attempts, that of ``non-commutative common causes'' \cite{Hofer2012, Hofer2013}, fails to provide an explanation of quantum correlations.

Consider the situation where two correlated events $A$ and $B$ are space-like separated, as in the usual Bell scenario. Relativistic causal structure implies that $A$ and $B$ cannot be causally connected, and thus the R-PCC implies they must share a common cause. Bell, like Reichenbach, considered a common cause explanation for these events to mean that there exists a sufficient specification of events (denoted by $\lambda$) in the common past of $A$ and $B$ such that conditioned on $\lambda$ the joint probability of $A$ and $B$ factorises, i.e. $P(A,B|\lambda)=P(A|\lambda)P(B|\lambda)$. This assumption however leads to the Bell inequalities, which are violated by quantum theory and by numerous experiments to date~\cite{Aspect1999,Rowe2001}. Thus relativistic causal structure and the R-PCC in the formulation above cannot coexist.

One alternative is to drop the R-PCC, and reject the idea that all correlations need to be explained in terms of a common cause. This position has been defended, for example, by van Fraassen~\cite{vanFraassen1982}, who called it ``the Charybdis of realism''. Another alternative is to drop relativistic causal structure, as is done in Bohmian mechanics \cite{Bohm1952}. However, this approach leads to a number of other problems, such as explaining why the violation of relativity cannot be observed at the macroscopic level. Recently it has been shown that any causal theory for the Bell correlations that maintains the R-PCC must feature such a ``fine-tuning'' of the causal parameters in the theory \cite{Wood2012}, that is, any such causal model must ``hide'' some causal channels available at the underlying level from use by macroscopic agents in the world. 

On the other hand, an outright rejection of R-PCC seems to be too extreme; the principle has a wide range of applicability and is explicitly assumed in disciplines such as the study of causal discovery algorithms \cite{Pearl2000}. It would therefore be interesting to attempt to modify the R-PCC or delineate its limits of validity. In this vein, some authors have proposed to keep some form of the R-PCC, but modified so that it blocks the derivation of Bell inequalities. Approaches in this line include that of Leifer and Spekkens (L-S) \cite{Leifer2006, Leifer2011} and that of Hofer-Szab\'{o} and Vecserny\'{e}s (H-V) \cite{Hofer2012}. 

In the L-S approach, the relativistic causal structure in the usual Bell scenario implies a factorisation not of probabilities, but of certain mathematical objects called \emph{conditional quantum states}, which are obtained through the Choi-Jamiolkowski isomorphism between quantum channels and quantum states. Wood and Spekkens \cite{Wood2012} have suggested that this approach could lead to a way of maintaining some non-classical form of the R-PCC that could find use in a generalised version of causal discovery algorithms compatible with the quantum probability calculus. The L-S formalism, however, has some caveats \cite{Leifer2011} and does not fully achieve that goal in its current form.

The H-V approach defines a notion of `non-commutative common cause', which requires a factorisation of probabilities similar to the standard Reichenbach-Bell formulation, but blocks the use of the Law of Total Probability that would lead to a derivation of Bell inequalities. Here we will give an analysis of this program and argue that it fails to provide a common cause explanation of the Bell correlations. The reason it fails is that it lacks an essential aspect of the notion of common cause, namely that the common cause can causally produce the observed correlations. Moreover, as we will show, using the H-V definition \emph{any} quantum product state can count as a common cause for \emph{any} quantum correlations whatsoever. There is thus no causal or inferential connection between the purported common causes and their effects, which makes the H-V notion of non-commutative common cause untenable.

\paragraph{Overview of paper}
In Section~\ref{sec:background} we show how Bell's theorem can be derived from the Principle of Common Cause and other assumptions; we then describe the options for responding to this result. In Section~\ref{sec:H-V} we review some basic concepts of algebraic quantum field theory, and the definition of non-commutative common causes given by H-V. In Section~\ref{sec:results} we describe our main results, which show how the definition of non-commutative common causes can be trivially satisfied by any product state for any quantum correlations. Finally, in Section~\ref{sec:discussion} we discuss how other approaches, in particular the Leifer-Spekkens model, may evade these problems  \cite{Leifer2011}.

\section{Common causes and quantum theory}\label{sec:background}

In this Section we shall discuss the relationship between R-PCC and Bell's theorem. In fact, the relationship is rather subtle. Indeed, although Reichenbach both articulated R-PCC and investigated no-go theorems for quantum theory, he did not quite anticipate the connection between them that Bell established. Reichenbach's no-go theorem is outlined in his 1944 book \cite{Reichenbach1944}---some twenty years before Bell's theorem---where he proves a result that he calls the `Principle of Anomaly'. This states that any theory that gives a more complete description of a physical system than that given by the wave function will be subject to `causal anomalies'. Prima facie, this sounds strikingly similar to the statement of Bell's theorem. However, he arrived at this conclusion by considering single systems instead of space-like separated systems, and despite his principle of anomaly, it doesn't seem that Reichenbach realised that quantum theory would come to challenge his principle of common cause. That Reichenbach did not arrive at Bell's theorem even though he considered both common causes and causal anomalies in quantum theory goes to show how subtle and deep was Bell's insight.

To better explore the connection between the Principle of Common Cause, Bell's notion of Local Causality (LC) \cite{Seevinck2010}, and the multiple responses to the dilemma imposed by Bell's theorem, we will use a weaker definition of the R-PCC, that does not imply separability of probabilities, as follows.
\begin{definition}[Principle of Common Cause (PCC)]\label{def:PCC}
If two events $A$ and $B$ are correlated, i.e., if $P(A,B)>P(A)P(B)$, then either:
	\begin{enumerate}[(i)]
		\item$A$ and $B$ are directly causally connected, i.e. either $A$ causes $B$ or $B$ causes $A$, or
		\item$A$ and $B$ share a \emph{common cause} that explains the correlation.
	\end{enumerate}
\end{definition}

\begin{definition}[Factorisation of probabilities (FP)]\label{def:RBCC}
Two events have a \emph{common cause} if and only if there exists a sufficient specification of variables $\lambda$ corresponding to events in the common causal past of $A$ and $B$ such that conditioned on those variables the joint probability of $A$ and $B$ factorises:
\begin{equation}\label{eq:RBCC}
P(A,B|\lambda) = P(A|\lambda)P(B|\lambda) \,.
\end{equation}
\end{definition}
	
The above is also referred to as a \emph{screening-off condition}. The conjunction of \cref{def:PCC} and \cref{def:RBCC} is Reichenbach's Principle of Common Cause (R-PCC) \cite{Reichenbach1956}. What counts as the `common causal past' of the two events in question of course depends on the causal structure assumed in a theory.

\begin{definition}[Relativistic causal structure (RCS)]\label{def:RCS}
The following two conditions are satisfied:
\begin{enumerate}
	\item \textit{Relativistic space-time}: Events can be embedded in a single relativistic space-time;
	\item \textit{Consistency of causal structure}: causality is consistent with background space-time structure  (e.g. within a relativistic space-time, all causes of an event are to be found in the event's past light cone).
\end{enumerate}
\end{definition}

Assuming relativistic causal structure (RCS), then if $A$ and $B$ are space-like separated events, they cannot be directly causally connected, and the PCC implies that they must share a common cause. Reinchenbach's PCC and relativistic causal structure imply Bell's notion of Local Causality (LC)~\cite{Wood2012}, which is contradicted by experiments with entangled quantum systems. However, to obtain that contradiction, it is required that one assumes the validity of the Law of Total Probability for the common cause events, and this is the assumption that is dropped in the H-V common cause formalism. For completeness, we review below the derivation of a Bell inequality, making this assumption explicit.

Suppose that two agents, Alice and Bob, each have access to a device for which Alice can supply an input $x$ and receive an output $a$; likewise Bob inputs $y$ and outputs $b$. In order to derive a contradiction with local causality, it is sufficient to consider the case where $x,y,a,b\in\{0,1\}$. In general, such a device will be probabilistic, in that its operation will be described by a conditional probability distribution $P(a,b\,|\,x,y)$.

Such a device can exhibit correlated outputs, e.g. the device could display perfect correlations, given by $P(a=b\,|\,x,y)=1$ for all $x,y$. Now, PCC leads us to seek an explanation of such correlations in terms of either (i) a direct cause from Alice to Bob (or vice versa); or (ii) a common cause to the past of both Alice and Bob. But suppose that Alice and Bob are space-like separated from one another. Then RCS implies that Alice and Bob cannot influence one another, and option (i) is unavailable. Consider then option (ii). This corresponds to the existence of an event $\lambda$, whose influence on the outputs of both Alice's and Bob's results, {\it ex hypothesi}, from an interaction in their common past (see Figure \ref{fig:lhbox}). 
(The variable $\lambda$ is also called a \emph{latent variable}, in the terminology of Bayesian networks \cite{Wood2012}.)
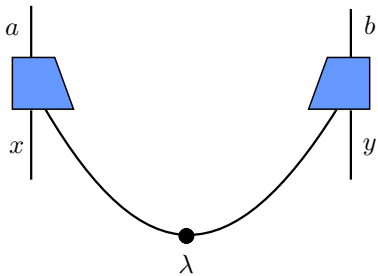
\begin{figure}
\centering
\ifx\JPicScale\undefined\def\JPicScale{1}\fi
\psset{unit=\JPicScale mm}
\psset{linewidth=0.3,dotsep=1,hatchwidth=0.3,hatchsep=1.5,shadowsize=1,dimen=middle}
\psset{dotsize=0.7 2.5,dotscale=1 1,fillcolor=black}
\psset{arrowsize=1 2,arrowlength=1,arrowinset=0.25,tbarsize=0.7 5,bracketlength=0.15,rbracketlength=0.15}
\begin{pspicture}(0,0)(48.12,36.25)
\rput(23.75,1.88){$\lambda$}
\rput(1.25,17.5){$x$}
\rput(48.12,17.5){$y$}
\rput(0.62,33.12){$a$}
\rput(48.12,33.75){$b$}
\newrgbcolor{userFillColour}{0.4 0.6 1}
\pspolygon[linewidth=0.2,fillcolor=userFillColour,fillstyle=solid](0.62,29.38)
(6.25,29.38)
(8.75,22.5)
(1.25,22.5)
(0.62,22.5)(0.62,29.38)
\newrgbcolor{userFillColour}{0.4 0.6 1}
\pscustom[linewidth=0.2,fillcolor=userFillColour,fillstyle=solid]{\psline(42.5,29.38)(48.12,29.38)
\psline(48.12,29.38)(48.12,22.5)
\psline(48.12,22.5)(40,22.5)
\psbezier(40,22.5)(40,22.5)(40,22.5)
\psline(40,22.5)(42.5,29.38)
\closepath}
\psline(3.12,36.25)(3.12,29.38)
\psline(45.62,35.82)(45.62,29.38)
\psline(3.12,22.34)(3.12,13.12)
\psline(45.62,22.34)(45.62,13.12)
\psbezier(5,22.5)(18.12,0.19)(29.75,0.19)(43.75,22.5)
\psdots[linestyle=none,fillstyle=solid,dotsize=0.7 4.5,dotstyle=o](23.75,5.62)
(23.75,5.62)
\end{pspicture} 
\caption{\small Schematic diagram of a Bell-type experiment. Alice and Bob have access to devices with inputs $x$ and $y$, and outputs $a$ and $b$, respectively. Any correlations between the outputs are sought to be explained by a common cause represented by a variable $\lambda$ in the common past of Alice and Bob's devices.}\label{fig:lhbox}
\end{figure}
Using the Reichenbach-Bell notion of common cause, Def.~\ref{def:RBCC}, to this setup means that conditioning on $\lambda$ factorises the correlations:
\beq\label{eq:factor}
P(a,b\,|\,x,y,\lambda)=P(a\,|\,x,\lambda)P(b\,|\,y,\lambda)
\eeq
Now comes the crucial step that is dropped in the H-V approach, which can be stated explicitly as:
\begin{definition}[Law of Total Probability (LTP)]\label{def:LTP}
Given a set of mutually exclusive events $\{x\}$ whose probabilities sum to unity, then the \emph{law of total probability} is satisfied if the probability of an event $y$ can be written as
\begin{equation} \label{eq:LTP}
P(y) = \sum_x P(x) P(y|x) \,.
\end{equation}
\end{definition}

Assuming the LTP, it must be that, when averaging over our ignorance of $\lambda$ using a probability distribution $\mu(\lambda)$, we obtain\footnote{Here another assumption is made, that the distribution of $\lambda$ does not depend on the free variables $x$ and $y$, i.e. $\mu(\lambda|x,y)=\mu(\lambda)$. We will leave this implicit in the main text as it is not important for our main argument.}:
\beq\label{eq:locality}
	P(a,b\,|\,x,y)=\sum_\lambda \mu(\lambda)\,P(a\,|\,x,\lambda)P(b\,|\,y,\lambda)
\eeq

Now, let $C_{xy}:=P(a=b\,|\,x,y)-P(a\neq b\,|\,x,y)$.
Bell showed that a set of probability distributions $P(a,b\,|\,x,y)$ satisfies Eq.~\eqref{eq:locality} (for all values of $a,b,x,y$) only if a certain inequality is satisfied \cite{Bell1964}, here we use the CHSH version \cite{Clauser1969}:
\beq \label{eq:Bell_inequality}
C:=C_{00}+C_{01}+C_{10}-C_{11}\leq 2.
\eeq
In quantum theory, there exist states and measurements which violate Eq.~\eqref{eq:Bell_inequality}; thence the resulting distributions $P(a,b\,|\,x,y)$ cannot be decomposed as in Eq.~\eqref{eq:locality} for any distribution over the variable $\lambda$. Moreover, Alice and Bob can be spacelike-separated when performing such measurements. For example, if Alice and Bob share the Bell state $\ket{\phi}=\frac{1}{\sqrt{2}}(\ket{00}+\ket{11})$, and choose appropriate measurements, then a violation of Eq.~\eqref{eq:Bell_inequality} will occur, with $C=2\sqrt{2}$~\cite{Clauser1969}. Furthermore, violation of Eq.~\eqref{eq:Bell_inequality} (and other variants) has been observed numerous times~\cite{Aspect1999,Rowe2001}. Thus we reach a contradiction, and one or more of our assumptions must be wrong.

In summary, Bell's theorem can be reformulated as stating that 
\begin{equation}
\mathrm{PCC + FP + RCS + LTP} \Longrightarrow \mathrm{Bell}\,\,\,\mathrm{inequalities}.
\end{equation}
Since quantum theory violates the Bell inequalities, this leaves us with the following alternatives: 
\begin{enumerate}[(a)]
	\item Reject the PCC, a position that has been defended by van Fraassen~\cite{vanFraassen1982}, among others; 
	\item Reject RCS, as in Bohmian mechanics \cite{Bohm1952} (which rejects condition (ii) in Definition~\ref{def:RCS});
	\item Reject FP, i.e. the notion that common causes require factorisation of probabilities, as in the approach of Leifer and Spekkens \cite{Leifer2006, Leifer2011}; or
	\item Reject LTP, as in the programme of Hofer-Szab\'{o} and Vecserny\'{e}s \cite{Hofer2012, Hofer2013}.
\end{enumerate}
The idea of dropping the LTP may sound plausible. After all, it is a well known fact that the LTP does not hold in general for counterfactual measurements in quantum theory\footnote{ For example, consider $x$ to be a which-path measurement on a double-slit experiment, and $y$ a measurement of the position a particle hits a screen after the slits. Then Eq.~\eqref{eq:LTP} represents the distribution one obtains when one measures the path the particle takes. This would not be compatible with the predictions for when such a measurement is \emph{not} performed, which would display quantum interference between the paths. Interestingly, a modified version of it holds in relation to a special kind of counterfactual measurement (SIC-POVMs). See \cite{Fuchs2009} for details.}, a fact encapsulated by Peres in the dictum ``unperformed experiments have no results'' \cite{Peres1978}. Since the hidden variables $\lambda$ are not directly measurable, then perhaps the LTP does not need to hold for them. 

On the other hand, the LTP does hold for any \emph{actual} events, such as the outcomes of performed measurements. Dropping the LTP is tantamount to saying that $\lambda$ have not actually occurred, in which case it is dubious whether they are the sorts of things that could count as a common cause for actual events. We will return to this argument later. First we will present what we take to be a much stronger argument against the `non-commmutative common cause' notion of H-V.

\section{Algebraic quantum field theory and non-com-mutative common causes}\label{sec:H-V}

The work by Hofer-Szab\'o and Vecserny\'es is formulated in the setting of algebraic quantum field theory (AQFT). To describe their definition of common cause we will need a few background notions. There are two features of AQFT to note. Firstly, AQFT is intended to be an operational approach to quantum field theory. This means that the primitives of the theory are observables, which are taken to be elements of a $C^*$-algebra.\footnote{For technical reasons, it is usually further assumed that this algebra is a von Neumann algebra. However, the distinction will not play a role in what follows.} Note that no initial reference is made to states or Hilbert spaces, as is consistent with the operational motivation, hence the use of `abstract' algebras of observables. Secondly, and relatedly, AQFT is sometimes referred to as a  `local' approach to quantum field theory. But this, of course, should not be taken to mean `local' in the sense of satisfying Bell-type inequalities. Instead, the AQFT concept of `locality' is formalised by assigning a $C^*$-algebra $\mathcal{A}(O)$ to each bounded region $O$ of a given spacetime, representing the observables that can be measured in $O$.

The basic initial structure of an AQFT is the mapping $O\mapsto\mathcal{A}(O)$, called a \emph{net of algebras}.\footnote{For details of further properties required of this mapping, and other aspects of AQFT, see \cite{Halvorson2006} for a full and accessible account of AQFT.} This mapping is required to satisfy various properties that formalise the motivations given above, and which ensure that the existence of an algebra $\hat{\mathcal{A}}$, known as the \emph{quasi-local algebra}.\footnote{The algebra $\hat{\mathcal{A}}$ is defined to be a direct limit of the net of algebras. See \cite{Rotman2010} for a full definition of this construction.} There exists an embedding $\mathcal{A}(O)\hookrightarrow\hat{\mathcal{A}}$ for every region $O$, and so the quasi-local algebra $\hat{\mathcal{A}}$ should be thought of as a `large' algebra encompassing all the local algebras $\mathcal{A}(O)$ at each region $O$.

To describe empirical data, we must have a notion of the state of a system. A \emph{state} is a positive linear functional $\phi:\hat{\mathcal{A}}\rightarrow\mathbb{C}$ on the quasi-local algebra $\hat{\mathcal{A}}$, meaning that $\phi(A)\in\mathbb{R}^+$ if $A$ is positive (i.e.~if $A=B^*B$ for some $B\in\mathcal{A}$). The quantity $\phi(A)$ yields the expectation value of the observable $A$ in the state $\phi$. The normalisation condition for states is $\phi(\mathbf{1})=1$, where the unit element is $\mathbf{1}\in\hat{\mathcal{A}}$.

We have defined states and observables, and so we can now begin to formulate the common cause condition of Hofer-Szab\'o and Vecserny\'es. The set of projections in the algebra $\mathcal{\hat{A}}$ is denoted $\mathcal{P}(\mathcal{\hat{A}})$. We call a set of mutually orthogonal projections $\{C_k\}_k$ a \emph{partition of unity} if $\sum_k C_k=\mathbf{1}$. This represents the set of mutually exclusive outcomes of a maximal measurement.

Given this framework, the authors of Ref.~\cite{Hofer2013} define a non-classical notion of common cause in the following way.
First, define the \emph{conditional expectation} $E_c$ as
	\begin{equation}
	E_c(A):=\sum_{k \in K} C_k A C_k \,. 
	\end{equation}
Then:
\begin{definition}\label{def:ccsystem}
A partition of the unit $\{C_k\}_{k\in K} \subset \mathcal{P(\hat{A})}$ is said to be the \emph{common cause system} of the commuting events $A, B \in \mathcal{P(\hat{A})}$, which correlate in the state $\phi: \mathcal{\hat{A}}\rightarrow \mathbb{C}$, if for all $k \in K$ such that $\phi(C_k)\neq 0$, the following condition holds:
	\begin{equation}\label{eq:HV-CC}
		\frac{(\phi \circ E_c)(ABC_k)}{\phi(C_k)}=\frac{(\phi \circ E_c)(AC_k)}{\phi(C_k)}\frac{(\phi \circ E_c)(BC_k)}{\phi(C_k)} \,.
	\end{equation}
\end{definition}
To unfold this definition, we use the L\"uders rule for expressing conditional probability in quantum theory \cite{Bub1977}. In AQFT this is given as follows. Let us write $P_\phi(A|C)$ to denote the probability, for a state $\phi$, of observing event $A$ conditioned on having observed event $C$. Then the L\"uders rule is:
\[
P_\phi(A|C):=\frac{\phi(CAC)}{\phi({C})}.
\]
Now, consider Definition~\ref{def:ccsystem}: it is easy to see that $E_c(AC_{k'})=C_{k'} A C_{k'}$, and so Eq.~\eqref{eq:HV-CC} can be rewritten as
	\begin{equation}
	\frac{\phi(C_kABC_k)}{\phi(C_k)}=\frac{\phi(C_kAC_k)}{\phi(C_k)}\frac{\phi(C_kBC_k)}{\phi(C_k)} \,.
	\end{equation}
Using the L\"uders rule, Eq.~\eqref{eq:HV-CC} corresponds to 
\[
P_\phi(A,B|C_k)=P_\phi(A|C_k)P_\phi(B|C_k).
\]
Hence, Eq.~\eqref{eq:HV-CC} indeed can be seen as expressing the screening-off condition for a common cause $C_k$, as in Eq.~\eqref{eq:RBCC}.

AQFT is an interesting approach to axiomatising quantum field theory. But we can analyse the notion of common cause given by Definition~\ref{def:ccsystem} in a simpler way---using ordinary quantum mechanics---as we now show.

\section{Non-commutative common causes are trivial}\label{sec:results}

To illustrate the physical content of Definition~\ref{def:ccsystem}, it will be helpful to assume that the algebra $\hat{\mathcal{A}}$ can be represented as a collection of operators on a Hilbert space. For example, $\hat{\mathcal{A}}$ could be a subset of the set of bounded linear operators $B(\mathcal{H})$ on a Hilbert space $\mathcal{H}$. In that case, a state $\phi$ is given by $\mathrm{Tr}[\rho-]$, where $\rho$ is an arbitrary density matrix. Then Eq.~\eqref{eq:HV-CC} becomes:
\begin{equation}\label{eq:cc_traces}
\frac{\mathrm{Tr}[\rho \, C_kABC_k]}{\mathrm{Tr}[\rho \, C_k]}=\frac{\mathrm{Tr}[\rho \, C_kAC_k]}{\mathrm{Tr}[\rho \, C_k]}\frac{\mathrm{Tr}[\rho \, C_kBC_k]}{\mathrm{Tr}[\rho \, C_k]} \,.
\end{equation}
As before, the term $\mathrm{Tr}[\rho \, C_kABC_k]$ can be recognised as expressing the probability $P_\rho(ABC_k)$ of obtaining a measurement outcome corresponding to $C_k$ followed by a measurement outcome corresponding to $AB$, given that the system was initially prepared in state $\rho$, and similarly for the other terms. The left-hand side thus represents the conditional probability $P_\rho(A,B\,|\,C_k)$ and the equation expresses the condition that the probability of obtaining $AB$ factories given that $C_k$ is observed, following preparation of state $\rho$:
\begin{equation}
P_\rho(A,B\,|\,C_k) = P_\rho(A|\,C_k) P_\rho(B\,|\,C_k)\,.
\end{equation}

We now proceed to show that any product state will satisfy the definition of a common cause system. We first require the following lemma.
\begin{lem}\label{thm:oneC}
Let $\hat{\mathcal{A}}\subseteq B(\mathcal{H})$.
Any projector corresponding to a product state, $C_k=\Pi_k^A \otimes \Pi_k^B$, where $\Pi_k^A$ and $\Pi_k^B$ are projectors onto the Hilbert spaces of subsystems $A$ and $B$, respectively, satisfies Eq.~\eqref{eq:HV-CC}.
\end{lem}
\begin{proof} 
As we have shown above, Eq.~\eqref{eq:HV-CC} can be rewritten as Eq.~\eqref{eq:cc_traces}. Substituting $C_k=\Pi_k^A \otimes \Pi_k^B$ in $\mathrm{Tr}[\rho \, C_kABC_k]$, we obtain
\begin{equation}\label{eq:defTr}
\mathrm{Tr}[\rho \, C_kABC_k] = \sum_{i,j,i',j'} \rho^{i'j'}_{ij} \langle i'|\Pi_k^A A \Pi_k^A|i\rangle \langle j'|\Pi_k^B B \Pi_k^B|j\rangle\,,
\end{equation}
where $\rho^{i'j'}_{ij}=\langle i|_A \langle j|_B \rho |i'\rangle_A |j'\rangle_B$ are matrix elements of $\rho$ relative to orthonormal bases $\{|i\rangle_A\}_i$ and $\{|j\rangle_B\}_j$. 
If we choose these bases so that $\Pi_k^A=|0\rangle_A\langle 0|$ and $\Pi_k^B=|0\rangle_B \langle 0|$, then Eq.~\ref{eq:defTr} simplifies to
\begin{equation}\label{eq:TrABCk}
\mathrm{Tr}[\rho \, C_kABC_k] = \rho^{00}_{00} \langle 0| A |0\rangle \langle 0|B |0\rangle\,.
\end{equation}
Similar reasoning shows that
\begin{eqnarray}\label{eq:TrACk}
\mathrm{Tr}[\rho \, C_kAC_k] &=& \rho^{00}_{00} \langle 0| A |0\rangle\,,\\
\mathrm{Tr}[\rho \, C_kBC_k] &=& \rho^{00}_{00} \langle 0| B |0\rangle\,,\label{eq:TrBCk}\\
 &\mathrm{and}& \nonumber\\
\mathrm{Tr}[\rho C_k]&=&\rho^{00}_{00}.\label{eq:denom}
\end{eqnarray}
Using Eqs.~\eqref{eq:TrABCk}--\eqref{eq:denom},
we see that Eq.~\eqref{eq:cc_traces} is satisfied.
\end{proof}

Using Lemma~\ref{thm:oneC}, we obtain the following.
\begin{thm}\label{thm:mainthm}
Let $\hat{\mathcal{A}}\subseteq B(\mathcal{H})$.
Any orthonormal set of product states for subsystems $A$ and $B$ forms a common cause system for the correlations of any quantum state.
\end{thm}
\begin{proof} 
Let $\{\Pi^A_k\}_k$ and $\{\Pi^B_k\}_k$ be arbitrary orthonormal bases of projectors for subsystems $A$ and $B$ respectively. The set of product states $\{C_k\}_k$, where $C_k:=\Pi^A_k\otimes \Pi^B_k$ satisfies $\sum_k C_k = \mathbb{I}$, and is thus a partition of the unit as required by Definition~\eqref{def:ccsystem}. We can then apply Lemma~\ref{thm:oneC} to $C_k$ for each $k$, by setting $\Pi^A_k=|k\rangle_A\langle k|$ and $\Pi^B_k=|k\rangle_B\langle k|$. Hence Eq.~\eqref{eq:HV-CC} is satisfied by $C_k$ for each $k$. Moreover, since the state $\rho$ in Lemma~\ref{thm:oneC} is arbitrary, $\{C_k\}_k$ is a common cause system for the correlations of any state $\rho$.
\end{proof}

Hence we see that any product state qualifies as a potential `non-commutative common cause' under the proposed definition. 
So, for example, suppose that Alice and Bob share a Bell state and perform measurements that lead to maximal violation of a Bell-type inequality. 
Then we are invited to model the common cause of these correlations as a product state in the common past of Alice and Bob---and \em any \em product state will do.
In addition, as Theorem~\ref{thm:mainthm} makes clear, this notion of common cause is doubly `degenerate'. For not only will any product state serve as a common cause, but it will do so for \emph{any} observed quantum correlations whatsoever.

 Importantly, since any product state can serve as a `non-commutative common cause' for any correlations, the actual observed correlations cannot in general be described as an expectation value over a distribution of possible common causes, using the Law of Total Probability, Eq.~\eqref{eq:LTP}, i.e.~we have:
\[
P_\rho(A,B)\neq \sum_k P_\rho(A,B|C_k)P_\rho(C_k).
\]
for any non-commutative common cause system $\{C_k\}_k$.
Indeed, take the $C_k$'s to be any set of product states. Then averaging over a distribution over those states, we obtain Eq.~\eqref{eq:locality}, and thus those correlations must necessarily satisfy all Bell inequalities. To maintain consistency with the predictions of quantum theory for a state that does violate a Bell inequality, the H-V approach must drop the LTP (as is explicitly acknowledged in \cite{Hofer2012,Hofer2013}).

\section{Discussion}\label{sec:discussion}

Now that we have an understanding of the H-V approach, it is instructive to return to the options in response to Bell's theorem that we have presented in Section~\ref{sec:background}.

We started out with a desire to maintain some version of the Principle of Common Cause in light of its apparent contradiction with quantum theory, as shown by Bell's theorem. 
We broke down the relevant assumptions that go into Bell's theorem into four parts, given by Defs.~\ref{def:PCC}-\ref{def:LTP}. These correspond to the following assumptions: (a) the Principle of Common Cause; (b) Factorisation of Probabilities; (c) Relativistic Causal Structure, and (d) the Law of Total Probability. We summarise the options available in Table~\ref{tbl:Bell_responses}.
\begin{table}[h!]
\caption{Responses to Bell's theorem, categorized according to assumptions satisfied}\label{tbl:Bell_responses} 
\centering 
\begin{tabular}{c | c c c c} 
\hline\hline 
Approach & PCC & FP & RCS & LTP \\ [0.5ex] 
\hline  
H-V & $\checkmark$ & $\checkmark$ & $\checkmark$ & $\times$ \\
L-S & $\checkmark$ & $\times$ & $\checkmark$ & $\checkmark$ \\ [1ex]
van Fraassen & $\times$ & $\times$ & $\checkmark$ & $\checkmark$ \\ 
Bohm & $\checkmark$ & $\checkmark$ & $\times$ & $\checkmark$ \\
\hline 
\end{tabular} 
\label{table:nonlin} 
\end{table} 

Let us take stock of the options displayed in Table~\ref{tbl:Bell_responses}.

\subsection*{The H-V response} 
We have seen that the approach of Refs.~\cite{Hofer2012,Hofer2013} avoids the contradiction by dropping the Law of Total Probability. However, we have shown that this leaves them with a notion of common cause that is too weak to satisfy an important purpose of Reichenbach's Principle of Common Cause itself, namely that the common cause is supposed to \emph{explain} the correlations. For in their definition, \emph{any} product state counts as a common cause for \emph{any} correlations. There is thus no deductive or causal link between the common cause and the effects they are purported to explain.

Let us be more specific about how this is manifested in the H-V approach, which is in two ways. Firstly, the H-V definition lacks one of the two crucial conceptual parts of the definition of a common cause. The definition of common cause systems given in Definition~\ref{def:ccsystem} provides a notion of conditional independence---an analogue of the factorisability condition for the commutative case---but it has no analogue of the assumption that the common causes can explain the correlations by averaging over the values of the hidden variables. In other words, it provides an analogue of Eq.~\eqref{eq:factor}, but not of Eq.~\eqref{eq:locality}. Now, even though it is true that LTP does not hold for counterfactuals in quantum theory (`unperformed experiments have no results'), it does hold for actual events. A cause must be an actual event, thus the LTP should presumably hold for them. 

Secondly, even granting that the LTP may be dropped, the notion of non commutative common cause is trivial in the sense that any product state can be posited as a common cause for any correlations. But no product state can reproduce the correlations that violate a Bell inequality. Their condition can thus be better interpreted as picking out states such that {\it if} we were to perform measurements over them, {\it then} they would render the variables independent, but this counterfactual does not in any way explain the correlations that we do observe when we have an entangled state. To explain the correlations of an entangled state by pointing out that a product state would factorise the correlations is analogous to trying to explain the simultaneous occurrence of a rare disease in opposite parts of the world by pointing out that the correlation \emph{could have not occurred}.

 Indeed this means that Eq.~\eqref{eq:LTP} \emph{cannot} be substituted by any function in the HV approach: the observed correlations are not a function of their common causes if all product states are common causes for all possible correlations. The correlations cannot be deduced, even probabilistically, from their causes. 

\subsection*{The L-S response}
We can compare this situation with the formalism of Leifer and Spekkens (L-S)~\cite{Leifer2011}. The authors of \cite{Wood2012} conjecture that one could retain a form of the PCC through a theory of causal networks such as that recently developed by the authors of~\cite{Leifer2011}, based on a notion of quantum conditional states proposed by Leifer~\cite{Leifer2006}. Their idea is to maintain the implications of the underlying causal structure to conditional independences, but replacing the classical conditional probabilities by quantum conditional states, i.e.~with a non-commutative generalisation of classical probability.

This therefore looks very similar to the H-V approach, so let us consider the L-S approach in more detail. L-S seek to define quantum theory as a \emph{formal analogue} of classical probability theory. For example, the standard density operator $\rho_A$ on a Hilbert space $\mathcal{H}_A$ is taken to play the role of a classical probability distribution $P(A)$.
A joint probability distribution and a marginal probability distribution are respectively replaced by a trace-one operator\footnote{
	In the case that systems $A$ and $B$ represent \emph{acausally} related systems (e.g.~spacelike separated), then $\sigma_{AB}$ is a bona fide density operator, i.e.~it is a \emph{positive} trace-one operator.
	However, in order to define quantum states for the case where $A$ and $B$ are causally related systems, L-S use the more general class of operators, which need not be positive.	
}
$\sigma_{AB}$ on a Hilbert space $\mathcal{H}_{AB}$, and 
the partial trace $\rho_A=\tr_B{[\sigma_{AB}]}$.
So far this is just a linguistic change to the usual formalism of quantum theory.
But the crucial new mathematical ingredient is given by a \emph{conditional density operator} $\sigma_{B|A}$ which is the analogue of a conditional probability distribution $P(B|A)$.
In particular, a conditional density operator satisfies
\beq\label{eq:cdo}
\sigma_{B|A}=\sigma_{AB}\star\rho_A^{-1}, 
\eeq
where the operation denoted by $\star$ is meant to provide an analogue of multiplying classical probability distributions (see \cite[p.~3]{Leifer2011} for an explicit definition of this operation).
Hence \cref{eq:cdo} is formally analogous to the definition of conditional probability, i.e.~$P(B|A)=P(A,B)/P(A)$.

In this way, L-S build up an entire theory of quantum systems that exactly mimics classical probability theory, including a quantum version of Bayes' theorem.
What, then, takes the role of LTP?
This corresponds to what L-S call \emph{belief propagation}: given a state $\rho_A$ for system $A$, and a conditional density operator $\sigma_{B|A}$, we can calculate the state for system $B$ as
\beq\label{eq:belief_prop}
\rho_B = \tr_A(\sigma_{B|A}\rho_A).
\eeq
Using this equation, the L-S approach can provide its own version of a common cause, as follows.
First, note that the formalism can also represent classical probability theory, since this can embed into quantum theory using density operators that are diagonal with respect to a basis representing a classical random variable $X$.
Using this idea, L-S define conditional density operators which can have a classical variable as output, denoted by $\rho_{X|A}$.	
Then L-S show that the Born rule is just an application of \cref{eq:belief_prop}.
That is, if we use standard quantum theory, then for a POVM $\{E^A_y\}_y$ and a density matrix $\rho_A$, we would write the Born rule as $P(Y=y)=\tr_A(E^A_y\rho_A)$. 
L-S show that this can instead be written as $\rho_Y=\tr_A(\sigma_{Y|A}\rho_A)$, for some conditional density operator $\sigma_{B|A}$, and where $\rho_Y$ encodes a classical probability in the manner just described.

Moreover, this then leads to the L-S version of a common cause, and to an important difference between the H-V and L-S approaches: the latter can correctly reproduce the correlations of states that violate a Bell inequality. 
More specifically, in the L-S formalism, the analogue of the locality condition \cref{eq:locality} is:
\beq
\rho_{XY}=\tr_{AB}\left((\sigma_{X|A}\sigma_{Y|B})\rho_{AB}\right)
\eeq
Hence the formal analogy between the L-S formalism and classical probability theory leads to the notion that $\rho_{AB}$ is a `common cause' for the correlations.
However, to reproduce the actual probabilistic correlations, L-S have replaced factorisation of probabilities (FP) with just the fact that the \emph{quantum channels} of the subsystems undergoing the Bell test factorise.
This type of factorisation leaves the L-S approach with room to use the initial (entangled) quantum state of the joint system to obtain the observable probabilities---but this comes at the cost of denying FP.

A similar way of dropping FP while keeping PCC would be to point out that correlations do not need to be explained in terms of a factorisability condition, but that the quantum state of the joint system in its causal past can \emph{itself} be considered as the common cause of the correlations. An objection to this point of view, however, is that the precise correlations cannot be determined without knowledge of the measurements to be performed (the inputs $x$ and $y$ in Fig.~\ref{fig:lhbox}), and these may be determined by factors not in the common past of the correlated events. A similar criticism may be made of the L-S approach. However, an advantage of the latter is that it does give an \emph{analogue} of the factorisation condition (rather than simply dropping it), and thus could allow for a generalisation of Reichenbach's Principle of Common Cause in understanding the implication of causal structure for probabilistic correlations, and be of potential application in areas such as causal discovery algorithms.

\section{Conclusion}

Our results emphasise the fact that the Principle of Common Cause plays a subtle and yet central role in the derivation of Bell's theorem. The subtlety arises from that fact that, a priori, one need not deny PCC in response to Bell's theorem. Yet, despite this option, we have seen that retaining PCC faces difficulties in practice.

Let us re-trace our path. We began by providing a taxonomy of the possible responses to Bell's theorem according to their relationship to PCC. We then showed that one particular response---that of Hofer-Szab\'{o} and Vecserny\'{e}s---faces a severe problem. In particular, our main result, Theorem~\ref{thm:mainthm}, showed that the H-V approach has a notion of common cause which is trivial. This triviality is two-fold: firstly, \emph{any} product state is a non-commutative common cause for some given set of correlations; secondly, any product state is a such a cause for \emph{any} correlations obtained by measurements on a bipartite quantum state. Furthermore, the notion that the common causes can in any way causally reproduce or explain the correlations is lost in this approach.

We then discussed how the H-V approach fares against other approaches, in particular the Leifer-Spekkens approach, which can reproduce the correlations, but at the cost of replacing the assumption of factorisation of probabilities by a different factorisation assumption, of conditional quantum states.

Our work suggests a few directions for future research, and we can mention two clusters of questions. Firstly, it would be interesting to understand the role of PCC and causality in other aspects of quantum theory. This could consist of providing a similar analysis for no-go theorems other than Bell's theorem. For example, it would be interesting to establish whether notions of causality can play a similar role in the derivation of the Kochen-Specker no-go theorem \cite{Kochen1967} for non-contextual hidden variable theories. 
Secondly, it would be interesting to understand whether there is a formal relationship between the H-V approach and the L-S approach. This leads to the more ambitious question of whether the L-S approach (or some variant thereof) can fully replace the role of Reichenbach's principle and provide a satisfactory causal explanation of quantum correlations.

\paragraph{Acknowledgements.}
We thank an anonymous referee for helpful feedback on an earlier version of this paper.
We thank Jon Barrett and Gabor Hofer-Szab\'{o} for useful discussions during the development of this work. EGC received support from an Australian Research Council grant DE120100559. RL received support from the Sydney Centre for Foundations of Science, and the John Templeton Foundation.

\bibliographystyle{plain}

\begin{thebibliography}{10}

\bibitem{Aspect1999}
Alain Aspect.
\newblock Bell's inequality test: more ideal than ever.
\newblock {\em Nature}, 398(6724):189--190, 1999.

\bibitem{Bell1964}
John~S Bell.
\newblock {On the Einstein-Podolsky-Rosen paradox}.
\newblock {\em Physics}, 1(3):195--200, 1964.

\bibitem{Bohm1952}
David Bohm.
\newblock A suggested interpretation of the quantum theory in terms of
  ``hidden" variables. i.
\newblock {\em Physical Review}, 85(2):166, 1952.

\bibitem{Bub1977}
Jeffrey Bub.
\newblock {Von Neumann's projection postulate as a probability
  conditionalization rule in quantum mechanics}.
\newblock {\em Journal of Philosophical Logic}, 6(1):381--390, 1977.

\bibitem{Clauser1969}
John~F. Clauser, Michael~A. Horne, Abner Shimony, and Richard~A. Holt.
\newblock {Proposed experiment to test local hidden variable theories}.
\newblock {\em Phys. Rev. Lett.}, 23:880--884, 1969.

\bibitem{Fuchs2009}
Christopher~A Fuchs and Ruediger Schack.
\newblock {Quantum-Bayesian Coherence}.
\newblock 2009.
\newblock {arXiv:0906.2187}.

\bibitem{Halvorson2006}
Hans Halvorson and Michael M\"uger.
\newblock {Algebraic quantum field theory}.
\newblock 2006.
\newblock arXiv:math-ph/0602036.

\bibitem{Hofer2012}
G{\'a}bor Hofer-Szab{\'o} and P{\'e}ter Vecserny{\'e}s.
\newblock Noncommuting local common causes for correlations violating the
  {C}lauser-{H}orne inequality.
\newblock {\em Journal of Mathematical Physics}, 53:122301, 2012.

\bibitem{Hofer2013}
G{\'a}bor Hofer-Szab{\'o} and P{\'e}ter Vecserny{\'e}s.
\newblock Bell inequality and common causal explanation in algebraic quantum
  field theory.
\newblock {\em Studies in History and Philosophy of Science Part B: Studies in
  History and Philosophy of Modern Physics}, 44(4):404 -- 416, 2013.

\bibitem{Kochen1967}
S.~Kochen and E.~Specker.
\newblock The problem of hidden variables in quantum mechanics.
\newblock {\em Journal of Mathematics and Mechanics}, 17:59--87, 1967.

\bibitem{Leifer2006}
Matthew~S Leifer.
\newblock Quantum dynamics as an analog of conditional probability.
\newblock {\em Physical Review A}, 74(4):042310, 2006.

\bibitem{Leifer2011}
MS~Leifer and RW~Spekkens.
\newblock {Formulating quantum theory as a causally neutral theory of Bayesian
  inference}.
\newblock 2011.
\newblock arXiv:1107.5849.

\bibitem{Pearl2000}
Judea Pearl.
\newblock {\em Causality: models, reasoning and inference}, volume~29.
\newblock Cambridge University Press, 2000.

\bibitem{Peres1978}
Asher Peres.
\newblock Unperformed experiments have no results.
\newblock {\em American Journal of Physics}, 46(7):745--747, 1978.

\bibitem{Reichenbach1944}
H.~Reichenbach.
\newblock {\em Philosophic Foundations of Quantum Mechanics}.
\newblock University of California Press, 1944.

\bibitem{Reichenbach1956}
Hans Reichenbach.
\newblock {\em The direction of time}.
\newblock Berkeley, University of Los Angeles Press, 1956.

\bibitem{Rotman2010}
Joseph~J. Rotman.
\newblock {\em Advanced Modern Algebra}, volume 114.
\newblock American Mathematical Society, 2010.

\bibitem{Rowe2001}
Mary~A Rowe, David Kielpinski, V~Meyer, Charles~A Sackett, Wayne~M Itano,
  Christopher Monroe, and David~J Wineland.
\newblock Experimental violation of a {B}ell's inequality with efficient
  detection.
\newblock {\em Nature}, 409(6822):791--794, 2001.

\bibitem{Seevinck2010}
MP~Seevinck.
\newblock Can quantum theory and special relativity peacefully coexist?
\newblock 2010.
\newblock {arXiv:1010.3714}.

\bibitem{vanFraassen1982}
B.~C. van Fraassen.
\newblock {The Charybdis of realism: epistemological implications of Bell's
  inequality}.
\newblock {\em Synthese}, 52:25--38, 1982.

\bibitem{Wood2012}
Christopher~J Wood and Robert~W Spekkens.
\newblock {The lesson of causal discovery algorithms for quantum correlations:
  Causal explanations of Bell-inequality violations require fine-tuning}.
\newblock 2012.
\newblock {arXiv:1208.4119}.

\end{thebibliography}

\end{document}